\documentclass[pra,twocolumn,floatfix,aps,showpacs,amsfonts,10pt]{revtex4}
\usepackage{graphicx}
\usepackage{array, color}
\usepackage{amsthm}
\usepackage{amsmath}
\usepackage{amssymb}
\newtheorem{theorem}{Theorem}

\newtheorem{definition}[theorem]{Definition}

\newtheorem{proposition}[theorem]{Proposition}

\begin{document}
\title{Coherence measure: Logarithmic coherence number}
\author{Zhengjun Xi}
\email{xizhengjun@snnu.edu.cn}
\affiliation{College of Computer Science, Shaanxi Normal University, Xi'an 710062,
China}

\author{Shanshan Yuwen}
\affiliation{College of Computer Science, Shaanxi Normal University, Xi'an 710062,
China}
\date{\today}

\begin{abstract}
We introduce a measure of coherence, which is extended from the coherence rank via the standard convex roof construction, we call it the logarithmic coherence number.
This approach is parallel to the Schmidt measure in entanglement theory, We study some interesting properties of the logarithmic coherence number, and show that this quantifier can be considered as a proper coherence measure. We also find that the logarithmic coherence number can be calculated exactly for a large class of states. We give the relationship between coherence and entanglement in bipartite system, and our results are generalized to multipartite setting. Finally, we find that the creation of entanglement
with bipartite incoherent operations is bounded by the logarithmic coherence number of the initial system during the process.
\end{abstract}
%\eid{identifier}
%\pacs{}
\maketitle
\section{Introduction}
The fundamental that distinguish quantum states from classical states is quantum coherence, which is the most basic characteristic of quantum mechanics.
Quantum coherence plays an important role in the study of quantum information and quantum multipartite systems.
Baumgratz $et$ $al$ proposed a theoretical framework for quantitative study of quantum coherence from the perspective of resource theory~\cite{BaumgratzPRL14}. Various ways have been builded to develop the resource-theoretic framework for
understanding quantum coherence, we refer to~\cite{StreltsovRMP17,Hu17} for more discussions of resource theory of coherence.

Analogously to the Schmidt rank in entanglement theory~\cite{TerhalPRAR01,EisertPRA01}, Killoran $et$ $al$ presented a framework for the conversion of
nonclassicality (including coherence) into entanglement, they
introduced a concept of the coherence rank~\cite{KilloranPRL16}. A concept related to the coherence rank was also discussed by Levi and Mintert~\cite{LeviNJP14}.
Soon afterwards, Chin introduced a discrete coherence monotone named the coherence number, which is a generalization of the
coherence rank to mixed states~\cite{ChinPRA17}. Regula $et$ $al$ also discussed coherence number of mixed states, they presented a general formalism for the conversion of nonclassicality into multipartite
entanglement~\cite{RegulaNJP18}. Theurer $et$ $al$ employed a natural generalization of the coherent rank to
superposition with respect to a finite number of linear independent basis~\cite{TheurerPRL17}.
The coherence number is proved to be a discrete coherence monotone, but it is not a proper coherence measure because it does not satisfy convexity~\cite{ChinPRA17,ZhaoPRL18}. To resolve this issue,
in this paper, we try to extend the coherence rank to mixed states via the standard convex roof construction,
this approach is parallel to the Schmidt measure in~\cite{EisertPRA01}. We can prove that it is not only a coherence monotone but also a proper coherence measure.

The paper is organized as follows. In Sec.~\ref{sec:basic}, we review some basic concepts about the resource theory of coherence. In Sec.~\ref{sec:defination}, we discuss the coherence rank and give a new property about it.  In Sec.~\ref{sec:properties}, we introduce a new coherence measure, which is the so-called the logarithmic coherence number. Some interesting properties are given. In Sec.~\ref{sec:multi scen}, we focus on the relationship between coherence and entanglement in bipartite and multipartite settings. In Sec.~\ref{sec:entanglement}, we discuss how the interplay between coherence consumption and
creation of entanglement. We summarizes our results in Sec.~\ref{sec:conclusion}.

\section{Basic concepts of coherence measure}\label{sec:basic}
We introduce some concepts about coherence measure that can be used for our main results~\cite{BaumgratzPRL14,StreltsovRMP17,Hu17}.
Given a $d$-dimension Hilbert space $\mathcal{H}$ with a fixed orthogonal basis $\mathcal{O}=\{|i\rangle\}_{i=0}^{d-1}$, we denote the set of all density operators acting on $\mathcal{H}$ by $\mathcal{D}(\mathcal{H})$. The density operators which are diagonal in this fixed basis are called incoherent, we denote the set of all incoherent states by $\mathcal{I}$, and $\mathcal{I}\subset\mathcal{D}(\mathcal{H})$. Any incoherent state $\delta$ is of the form
\begin{equation}
\delta=\sum_{i=0}^{d-1}\delta_i|i\rangle\langle i|,
\end{equation}
where $\delta_i$ are probability distribution.
Any state which cannot be written as above form is defined as a coherent state,
which means the coherence is basis-dependent.

The incoherent operation is to map the incoherent states to incoherent states. The definition of incoherent operation is not unique and different choices~\cite{StreltsovRMP17}.
In this paper, we only consider the incoherent operation in~\cite{BaumgratzPRL14}. The incoherent operation (IO) is a completely positive and trace preserving (CPTP) maps $\Lambda$ that admit a Kraus operator representation
\begin{equation}
\Lambda(\rho)=\sum_nK_n\rho K_n^\dagger,
\end{equation}
where all the Kraus operators $K_n$ must satisfy $K_n \mathcal{I} K_n^\dagger\subseteq\mathcal{I}$ with $\sum_n K_n^{\dagger}K_n= I$. In general, the Kraus operator can always be represented as
\begin{equation}
K_{n}=\sum_{i}c_i|f(i)\rangle\langle i|,
\end{equation}
where $f$ is a function in the index set and $c_i \in [0,1]$~\cite{WinterPRL16}.

Baumgratz $et$ $al$ proposed that any proper measure of the coherence $\mathcal{C}$ must satisfy the following conditions~\cite{BaumgratzPRL14}:
\begin{itemize}
\item [{(C1)}]  $Nonnegativity:$ $\mathcal{C}(\rho)\geq 0$ for all quantum states $\rho$, and $\mathcal{C}(\rho)=0$ if and only if $\rho$ is incoherent.
\item [{(C2)}]  $Monotonicity:$ $\mathcal{C}(\rho)$ is non-increasing under incoherent operation $\Lambda$, i.e.,
$\mathcal{C}(\rho) \geq \mathcal{C}(\Lambda(\rho))$.
\item [{(C3)}]  $Strong$ $monotonicity\!:$ $\mathcal{C}(\rho)$ does not increase on average under selective incoherent operations, i.e., $\sum_n q_n\mathcal{C}(\rho_n)\leq\mathcal{C}(\rho)$, where $\rho_n=K_n\rho K_n^\dag/q_n$, and $q_n=\mathrm{Tr}(K_n \rho K_n^\dag)$.
\item [{(C4)}]  $Convexity:$ $\mathcal{C}(\rho)$ is a convex function of quantum states, i.e.,
%\begin{equation}
$\sum_ip_i \mathcal{C}(\rho_i)\geq \mathcal{C}(\sum_ip_i\rho_i),$
%\end{equation}
 for any ensemble $\{p_i,\rho_i\}$.
\end{itemize}
Following standard notions from entanglement theory, we call a quantifier $\mathcal{C}$ which fulfills conditions (C1) and either condition (C2) or (C3)~(or both) a coherence monotone. A quantifier $\mathcal{C}$ is further called a coherence measure if it satisfies the four conditions:~(C1)-(C4). We also know that conditions (C3) and (C4) automatically imply condition (C2)~\cite{StreltsovRMP17}.

\section{Coherence rank}\label{sec:defination}
For a pure state on Hilbert space $\mathcal{H}$ with the fixed orthogonal basis $\mathcal{O}$, one can define the coherence
rank
\begin{equation}\label{equ:ECS}
R_C(|\psi\rangle)=\min\left\{ |\hat{\mathcal{O}}| \mid|\psi\rangle=\sum_{|j\rangle\in\hat{\mathcal{O}}}\lambda_{j}|j\rangle, \hat{\mathcal{O}}\subseteq\mathcal{O}\right\},
\end{equation}
 where $\lambda_{j}$ are nonzero complex coefficients.

We note that the coherence
rank given in Eq.~(\ref{equ:ECS}) characterizes the minimal number of the incoherent states in the fixed orthogonal basis $\mathcal{O}$ in such a decomposition of $|\psi\rangle$. This is also equivalent to the fact that the coherence rank $R_C(|\psi\rangle)=k$ if exactly $k$ of the coefficients $\lambda_{j}$ are nonzero. Thus, we say that the definition of the coherence rank given in
Eq.~(\ref{equ:ECS}) is equivalent to the definition that was introduced in Refs.~\cite{KilloranPRL16,TheurerPRL17,ChinPRA17,LeviNJP14,RegulaNJP18}.
Clearly, we have $1\leq R_C(|\psi\rangle)\leq d$ and all coherent pure states should have $R_C(|\psi\rangle)\geq 2$. We know that the coherence
rank is non-increasing under incoherent operations $\Lambda$, that is,
\begin{equation}\label{equ:rkc}
R_{C}(\Lambda(|\psi\rangle))\leq R_{C}(|\psi\rangle).
\end{equation}
In particular, following the results in~\cite{KilloranPRL16,WinterPRL16}, we know that there exists a unitary incoherent operation $U_{\mathrm{in}}$ on a pure state $|\psi\rangle$ such that the coherence rank of $U_{\mathrm{in}}|\psi\rangle$ is equal to the coherence rank of $|\psi\rangle$, i.e.,
\begin{equation}
R_C(U_{\mathrm{in}}|\psi\rangle)=R_C(|\psi\rangle),
\end{equation}
where $U_{\mathrm{in}}=\sum_{j}e^{i\theta_j}|j\rangle\langle j|$ with some phases $\theta_j$.
Therefore, we say that the coherence rank is a coherence monotone.

We also consider the coherence rank of superposition of two coherent states. The following result will give the lower and upper bounds of the coherence of superposition.
\begin{proposition}
Let $|\phi\rangle=a|\psi\rangle + b|\varphi\rangle$ with $|a|^2+|b|^2=1$, we have
\begin{equation}
|R_{C}(|\psi\rangle)\!-\!R_{C}(|\varphi\rangle)|\!\leq\!R_{C}(|\phi\rangle)\!\leq\! R_{C}(|\psi\rangle)\!+\!R_{C}(|\varphi\rangle).
\end{equation}
\end{proposition}
\begin{proof}
By the definition of the coherence rank, there exist two sets $\hat{\mathcal{O}}_{\psi}$ and $\hat{\mathcal{O}}_{\varphi}$ such that
$R_{C}(|\psi\rangle)=|\hat{\mathcal{O}}_{\psi}|$, $R_{C}(|\varphi\rangle)=|\hat{\mathcal{O}}_{\varphi}|$, and one has
\begin{equation}
|\psi\rangle=\sum_{|j\rangle\in\hat{\mathcal{O}}_{\psi}}\psi_{j}|j\rangle,\ \ |\varphi\rangle=\sum_{|k\rangle\in\hat{\mathcal{O}}_{\varphi}}\varphi_{k}|k\rangle.
\end{equation}
Then, we will consider three cases as follows.

\emph{Case 1.} If $\hat{\mathcal{O}}_{\psi}\perp\hat{\mathcal{O}}_{\varphi}$, by definition, we directly obtain
\begin{equation}
R_{C}(\phi\rangle)= R_{C}(|\psi\rangle)+R_{C}(|\varphi\rangle).
\end{equation}

\emph{Case 2.} If $\hat{\mathcal{O}}_{\psi}\cap\hat{\mathcal{O}}_{\varphi}\neq \varnothing$, without loss of generality, we take $\tilde{\mathcal{O}}=\hat{\mathcal{O}}_{\psi}\cap\hat{\mathcal{O}}_{\varphi}$, and
\begin{equation}
|\phi\rangle\!=\!a\!\sum_{|j\rangle\in\hat{\mathcal{O}}_{\psi}\backslash\tilde{\mathcal{O}}}\psi_{j}|j\rangle+ \sum_{|j\rangle\in\tilde{\mathcal{O}}}(a\psi_{j}+b\varphi_{j})|j\rangle+b\sum_{|k\rangle\in \hat{\mathcal{O}}_{\varphi}\backslash \tilde{\mathcal{O}}}\varphi_{k}|k\rangle.
\end{equation}
Then, we have
\begin{align}
R_{C}(|\phi\rangle)\leq&\  |\hat{\mathcal{O}}_{\psi}\backslash \tilde{\mathcal{O}}|+|\hat{\mathcal{O}}_{\varphi}\backslash \tilde{\mathcal{O}}|+|\tilde{\mathcal{O}}|\nonumber\\
=&\ |\hat{\mathcal{O}}_{\psi}|+|\hat{\mathcal{O}}_{\varphi}|-|\tilde{\mathcal{O}}|\nonumber\\
\leq&\ R_{C}(|\psi\rangle)+R_{C}(|\varphi\rangle).
\end{align}

\emph{Case 3.} If $\hat{\mathcal{O}}_{\psi}\subseteq \hat{\mathcal{O}}_{\varphi}$, then we have
\begin{equation}
|\phi\rangle=\sum_{|j\rangle\in \hat{\mathcal{O}}_{\psi}}(a\psi_{j}+b\varphi_{j})|j\rangle+
b\sum_{|k\rangle\in\hat{\mathcal{O}}_{\varphi}\backslash\hat{\mathcal{O}}_{\psi}}\varphi_{k}|k\rangle,
\end{equation}
By the definition, we obtain
\begin{equation}
R_{C}(|\phi\rangle)\geq R_{C}(|\varphi\rangle)-R_{C}(|\psi\rangle).
\end{equation}
Similarly, If $\hat{\mathcal{O}}_{\varphi}\subseteq\hat{\mathcal{O}}_{\psi}$, we have
\begin{equation}
R_{C}(|\phi\rangle)\geq R_{C}(|\psi\rangle)-R_{C}(|\varphi\rangle).
\end{equation}
Thus, we obtain our desired result.
\end{proof}
The coherence rank has been generalized to mixed states in~\cite{ChinPRA17,KilloranPRL16,RegulaNJP18}, it is the so-called coherence number, which is defined as
 \begin{equation}
R_{C}(\rho)=\min_{\{p_i,|\psi_i\rangle\}}\max_i[R_{C}(|\psi_i\rangle)].
\end{equation}
The coherence number is the smallest possible maximal coherence rank in any decomposition of the mixed states, and for pure states the coherence rank equals the coherence number.
The coherence number only satisfies condition (C1), (C2) and (C3), but it doesn't satisfy condition (C4)~\cite{ChinPRA17,ZhaoPRL18}, so it is only a coherence monotone. In the following
section, we apply the standard convex roof construction to the mixed states.

\section{Logarithmic coherence number}\label{sec:properties}
In this section, we can define logarithmic coherence rank,
it is in the same way as the Schmidt rank in~\cite{EisertPRA01}.  Note that Theurer $et$ $al$ used this way to consider the superposition in~\cite{TheurerPRL17}.
\begin{definition}
For any pure state $|\psi\rangle$ on $\mathcal{H}$, the logarithmic coherence rank is defined as
\begin{equation}\label{Def:LCR}
\mathcal{L}_{C}(|\psi\rangle)=\log_{2}R_{C}(|\psi\rangle).
\end{equation}
\end{definition}
Obviously, the logarithmic coherence rank inherits some properties of coherence rank. The logarithmic coherence rank is non-negative, that is, $\mathcal{L}_C(|\psi\rangle)\geq 0$ for any pure state $|\psi\rangle$. In particular, for the maximally coherent states
\begin{equation}\label{eq:MCS}
|\psi_M\rangle=\frac{1}{\sqrt{d}}\sum_{j=0}^{d-1}e^{i\theta_j}|j\rangle,
\end{equation}
we have
\begin{equation}
\mathcal{L}_C(|\psi_M\rangle)=\log_2d.
\end{equation}
In addition, we find that the logarithmic coherence rank is also monotone, unitarily invariant and so on.
The logarithmic coherence rank can be extended to mixed states by the standard convex roof construction.
\begin{definition}\label{equ:definition 1}
For any mixed state $\rho$, the logarithmic coherence number is defined as
\begin{equation}
\mathcal{L}_{C}(\rho)=\min_{\{p_i,|\psi_i\rangle\}}\sum_ip_i\mathcal{L}_{C}(|\psi_i\rangle),\label{eq:LCN}
\end{equation}
where the minimum is taken over all pure state decompositions of $\rho=\sum_ip_i|\psi_i\rangle\langle\psi_i|$.
\end{definition}
In the subsequent paragraphs we will show that
the logarithmic coherence number is a proper coherence measure in the sense of Refs.~\cite{BaumgratzPRL14,StreltsovRMP17}.
\begin{proposition}
The logarithmic coherence number ${\mathcal{L}}_{C}$ is a coherence measure, which satisfies the conditions (C1)-(C4).
\end{proposition}
\begin{proof}
Obviously, condition (C1) follows immediately from the definition.

To show that $\mathcal{L}_C$ satisfies condition (C3), let $\rho=\sum_ip_i|\psi_i\rangle\langle\psi_i|$ be the optimal decomposition of $\rho$ belonging to the minimum in Eq.~(\ref{eq:LCN}), and we take the mark~\cite{TheurerPRL17}, and define
\begin{equation}
|\hat{\psi}_{i,n}\rangle=\frac{K_n|\psi_i\rangle}{\sqrt{q_n}},
\end{equation}
where $q_n=\mathrm{Tr}(K_{n}^{\dag}K_{n}\rho)$, and $K_n$ are incoherent Kraus operators. Then, every final state $\rho_n$ in an incoherent Kraus operator $K_n$  can be represented by
 \begin{equation}
\rho_n=\frac{K_{n}\rho K_{n}^{\dag}}{q_n}=\sum_ip_i|\hat{\psi}_{i,n}\rangle\langle\hat{\psi}_{i,n}|.
 \end{equation}
Since the coherence rank can never increase under
the action of an incoherent Kraus operator, then we have
\begin{align}
\mathcal{L}_{C}(\rho_n)&\leq\sum_ip_i\mathcal{L}_{C}(|\hat{\psi}_{i,n}\rangle)\nonumber\\
&\leq\sum_ip_i\mathcal{L}_{C}(|\psi_i\rangle)\nonumber\\
&=\mathcal{L}_{C}(\rho).
\end{align}
Thus, we have
\begin{equation}
\sum_nq_n\mathcal{L}_{C}(\rho_n)\leq\mathcal{L}_{C}(\rho).
\end{equation}

To show (C4) we take
\begin{equation}
\rho=\lambda_1\rho_{1}+\lambda_2\rho_{2},
\end{equation}
where $\lambda_1, \lambda_2\in[0,1]$.
Let $\rho_1=\sum_j\mu_j|\phi_j\rangle\langle\phi_j|$ and $\rho_2=\sum_k\eta_k|\varphi_k\rangle\langle\varphi_k|$ be the two decompositions for which the respective minima in Eq.~(\ref{eq:LCN}) are attained. Then the convex combinations $\lambda_1\sum_j\mu_j|\phi_j\rangle\langle\phi_j|+\lambda_2\sum_k\eta_k|\varphi_k\rangle\langle\varphi_k|$ is a valid decomposition of $\rho$, but it is not necessarily the optimal one. Thus, we have
\begin{align}
\mathcal{L}_C(\lambda_1 \rho_{1}\!+\!\lambda_2\rho_{2})
 \!\leq&\lambda_1\sum_{j}\mu_{j}\mathcal{L}_C(|\phi_{j}\rangle)\!+\!
 \lambda_2\sum_{k}\eta_{k}\mathcal{L}_C(|\varphi_{k}\rangle)\nonumber\\
\!\leq&\lambda_1\mathcal{L}_C(\rho_1)+\lambda_2\mathcal{L}_C(\rho_2).
 \end{align}

We know that the condition (C2) can be derived from conditions (C3) and (C4),
so we say the logarithmic coherence number $\mathcal{L}_{C}$ satisfies conditions (C1)-(C4).
\end{proof}

This shows that the logarithmic coherence number can indeed be used as a coherence measure quantifying the coherence of a quantum system.
Not just these nice properties, we also find that the logarithmic coherence number is additive as follows.
\begin{proposition}\label{prop:additive}
The logarithmic coherence number $\mathcal{L}_{C}$ is additive.
\end{proposition}
\begin{proof}
Let us consider the case of pure states
first. From the definition of the coherence rank, we have
\begin{equation}
R_{C}(|\psi_1\rangle\otimes|\psi_2\rangle)=R_C(|\psi_{1})\cdot R_C(|\psi_2\rangle).
\end{equation}
Thus, we obtain
\begin{align}
\mathcal{L}_{C}(|\psi_1\rangle\otimes|\psi_2\rangle)=&\log_2(R_{C}(|\psi_{1})\cdot R_{C}(|\psi_{2}\rangle)\nonumber\\
=&\mathcal{L}_{C}(|\psi_1\rangle)+\mathcal{L}_C|\psi_2\rangle).
\end{align}

Then we consider the case of mixed states. Without loss of generality,
the pure states decompositions of $\rho\otimes\sigma$ is of the form
 \begin{equation}
 \rho\otimes\sigma=\sum_ap_a|\psi_a\rangle\langle\psi_a|\otimes\sum_bp_b|\phi_b\rangle\langle\phi_b|.
 \end{equation}
Then we have
\begin{align}
\mathcal{L}_{C}(\rho\otimes\sigma)&=
\min\sum_{a,b}p_ap_b\mathcal{L}_{C}(|\psi_a\rangle\otimes|\phi_b\rangle)\nonumber\\
&=\min\sum_ap_a\mathcal{L}_{C}(|\psi_a\rangle)\otimes \min\sum_bp_b\mathcal{L}_C(|\phi_b\rangle)\nonumber\\
&=\mathcal{L}_{C}(\rho)+\mathcal{L}_{C}(\sigma).
\end{align}
This completes the proof the proposition.
\end{proof}

From this result, for $n$ copies of the same state $|\psi\rangle$, we have
\begin{equation}
\mathcal{L}_C(|\psi\rangle^{\otimes n})=n\mathcal{L}_C(|\psi\rangle).
\end{equation}
In particular, let $\delta$ be an incoherent state, we have
\begin{equation}
\mathcal{L}_C(\delta^{\otimes n}\otimes |\psi\rangle\langle\psi|^{\otimes n})=n\mathcal{L}_C(|\psi\rangle).
\end{equation}

If the states $|\psi_1\rangle$ and $|\psi_2\rangle$ satisfy $||\psi_1\rangle-|\psi_2\rangle|<\varepsilon$, we may ask whether the logarithmic coherence number also satisfies $|\mathcal{L}_{C}(|\psi_1\rangle)-\mathcal{L}_{C}(|\psi_2\rangle)|<\varepsilon$, where $|\cdot|$ is trace distance. Let $|\psi_1\rangle$ be the state
\begin{equation}
|\psi_1\rangle=\sqrt{1-\varepsilon}|0\rangle+\sqrt{\frac{\varepsilon}{d-1}}\sum_{i=1}^{d-1}|i\rangle,
\end{equation}
and $|\psi_2\rangle=|0\rangle$. When $\varepsilon\rightarrow0$, which means $|\psi_1\rangle\rightarrow|\psi_2\rangle$, but we know that $\big|\mathcal{L}_{C}(|\psi_1\rangle)-\mathcal{L}_C(|\psi_2\rangle) \big|=\log_2d$. Thus, we claim that the logarithmic coherence number is not continuous.

Although we define the coherence measure of a mixed state via a minimization over all possible realizations of the state, it can be calculated exactly for some states. In order to calculate the logarithmic coherence number of a mixed state, the minimization over decompositions of the state is necessary. The value of $\mathcal{L}_C$ can be fully evaluated for some states. We consider a family of noisy maximally coherent states
\begin{equation}
\rho_\lambda=\lambda|\psi_M\rangle\langle\psi_M|+(1-\lambda)\frac{I}{d},
\end{equation}
where $\lambda\in[0,1]$.
Without loss of generality, the identity operator $I$ can be represented with the pure states $|\psi_i\rangle$ as
\begin{equation}
I=\sum_i\alpha_i|\psi_i\rangle\langle\psi_i|,
\end{equation}
where $\alpha_i\geq0$. Then, the pure states decompositions of $\rho_\lambda$ is of the form
\begin{equation}
\rho_\lambda=\lambda|+\rangle\langle+|+
\frac{1-\lambda}{d}\sum_i\alpha_i|\psi_i\rangle\langle\psi_i|.
\end{equation}
%Since the logarithmic coherence number of every state $|\psi_i\rangle$ is (incoherent state) $0$ or (coherent state) $1$,
Using the definition~(\ref{eq:LCN}), we get
\begin{equation}\label{eq:1-qubit co}
\mathcal{L}_C(\rho_\lambda)\!=\!\!\min_{\{\alpha_i,|\psi_i\rangle\}}\!
\left[\!\lambda\log_2d\!+\!\frac{1-\lambda}{d}\!\sum_{i}\alpha_i\mathcal{L}_C(|\psi_i\rangle)\!\right]\!.
\end{equation}
Minimizing the right-hand side of Eq.~(\ref{eq:1-qubit co}) over all pure states decompositions we immediately see that the minimum is achieved for every $i$, $\mathcal{L}_C(|\psi_i\rangle)=0$.
Thus, we obtain a closed expression of the logarithmic coherence number for the state $\rho_\lambda$, i.e.,
\begin{equation}
\mathcal{L}_C(\rho_\lambda)=\lambda\log_2d.
\end{equation}
\section{Multipartite scenario}\label{sec:multi scen}
Let $\mathcal{H}^S$ and $\mathcal{H}^A$ be two $d$-dimensional Hilbert spaces, and $\mathcal{H}^A$ be the Hilbert space of an ancillary system with $\mathcal{H}^S\cong\mathcal{H}^A$. Without loss of generality, we take the orthogonal basis $\{|i\rangle\}_{i=0}^{d-1}$ and $\{|j\rangle\}_{i=0}^{d-1}$ as two fixed basis on $\mathcal{H}^S$ and $\mathcal{H}^A$, respectively. Then, their tensor product $\{|i\rangle\otimes |j\rangle\}$ can be viewed as an incoherent basis for compound system $SA$. Thus, the corresponding logarithmic coherence rank and the logarithmic coherence number can be defined just as~(\ref{Def:LCR}) and~(\ref{eq:LCN}). We are particularly interested in the relationship between the total coherence and coherence contained in each individual subsystem. In the following proposition, we prove that the logarithmic coherence number in the bipartite quantum states is no less than the sum between two subsystems. This relation can be viewed as the super-additivity for the logarithm coherence number.
\begin{proposition}
For any bipartite quantum state $\rho^{SA}$ on $SA$, we have
\begin{equation}\label{eq:super-addi}
\mathcal{L}_C(\rho^{S})+\mathcal{L}_C(\rho^{A})\leq\mathcal{L}_C(\rho^{SA}),
\end{equation}
where $\rho^S$ and $\rho^A$ are reduced states on $S$ and $A$, respectively.
\end{proposition}
\begin{proof}
Firstly, we consider the case of pure states. Let
\begin{equation}
|\psi^{SA}\rangle = \sum_{i=0}^{r_S-1}\sum_{j=0}^{r_A-1}a_{ij}|i^S\rangle|j^A\rangle.
 \end{equation}
be the optimal decomposition of $|\psi^{SA}\rangle$ belonging to the minimum in Eq.~(\ref{Def:LCR}), it follows that
\begin{equation}\label{eq:2pure LCN}
R_C(|\psi^{SA})=r_S\times r_A.
\end{equation}
Further, the matrix $M$ of complex numbers $a_{ij}$ can be represented as
\begin{equation}
M=\left(
    \begin{array}{cc}
      N_{r_S\times r_A} & O \\
      O & O \\
    \end{array}
  \right),
\end{equation}
where $N_{r_S\times r_A}=(a_{ij})_{r_S\times r_A}$, and $O$ are zero matrices.
Using the singular value decomposition, $M=U\Sigma V$, where $\Sigma$ is a diagonal matrix with non-negative elements $\lambda_m$,
which are the singular values of $M$, and $U$ and $V$ are unitary matrices. Thus, it is always possible to write $|\psi^{SA}\rangle$ in the following way
\begin{equation}
|\psi^{SA}\rangle=\sum_{m=0}^{r-1}\lambda_m|m^S\rangle|m^A\rangle,
\end{equation}
where $r$ is the Schmidt number of the state $|\psi^{SA}\rangle$, and
\begin{equation}\label{eq:reduced form11}
|m^S\rangle=\sum_{i=0}^{r_S-1}u_{im}|i^S\rangle, |m^A\rangle=\sum_{j=0}^{r_A-1}v_{mj}|j^A\rangle.
\end{equation}
Here the complex number $u_{im}$ and $v_{mj}$ are matrix elements of unitary matrices $U$ and $V$. It is easy to see that the coherence rank of the states $|m^S\rangle$ and $|m^A\rangle$ can not exceed the numbers $r_S$ and $r_A$, respectively. This means that for every $m$ the following inequalities holds,
\begin{equation}
\mathcal{L}_C(|m^S\rangle)\leq \log_2r_S,\mathcal{L}_C(|m^A\rangle)\leq \log_2r_A.
\end{equation}
For the subsystem $S$, we know that $\rho^S=\sum_m\lambda^2_m|m\rangle^S\langle m|$ is a valid decomposition of $\rho^S$,  then we obtain
\begin{align}
\mathcal{L}_C(\rho^S)& \leq \sum_m\lambda^2_m\mathcal{L}_C(|m^S\rangle)\nonumber\\
&\leq\sum_m\lambda^2_m\log_2r_S\nonumber\\
&=\log_2r_S,
\end{align}
and similarly that
 \begin{align}
\mathcal{L}_C(\rho^A)\leq \log_2r_A.
\end{align}
The above inequalities together with Eq.~(\ref{eq:2pure LCN}) implies the following inequality,
 \begin{equation}\label{eq:pure_LCN}
\mathcal{L}_C(\rho^S)+\mathcal{L}_C(\rho^A)\leq \mathcal{L}_C(|\psi^{SA}\rangle).
\end{equation}

For any mixed state $\rho^{SA}$, let $\rho^{SA}=\sum_ip_i|\psi_i\rangle^{SA}\langle\psi_i|$ be the optimal decomposition of $|\psi^{SA}\rangle$ belonging to the minimum in Eq.~(\ref{eq:LCN}), we have
\begin{equation}\label{eq:BMS_LCN}
\mathcal{L}_C(\rho^{SA})=\sum_ip_i\mathcal{L}_C(|\psi_i^{SA}\rangle).
\end{equation}
Combining Eq.~(\ref{eq:pure_LCN}) and Eq.~(\ref{eq:BMS_LCN}), we obtain
\begin{align}
\mathcal{L}_C(\rho^{SA})&=\sum_ip_i\mathcal{L}_C(|\psi_i^{SA}\rangle)\nonumber\\
&\geq \sum_ip_i\mathcal{L}_C(\rho^S_i)+\sum_ip_i\mathcal{L}_C(\rho^A_i)\nonumber\\
&\geq \mathcal{L}_C(\rho^S)+\mathcal{L}_C(\rho^A).
\end{align}
This completes the proof of the proposition.
\end{proof}

From the proof of the proposition,
we immediately see that the Schmidt number $r$ does not exceed to the numbers $r_S$ and $r_A$, i.e.,
\begin{equation}
r\leq \min\{r_S,r_A\}.
\end{equation}
Thus, we can obtain an interesting relation between entanglement and coherence as follows,
\begin{equation}
\max\{\mathcal{L}_C(\rho^S),\mathcal{L}_C(\rho^A)\}+\mathcal{L}_E(|\psi^{SA}\rangle)\leq \mathcal{L}_C(|\psi^{SA}\rangle),
\end{equation}
 where $\mathcal{L}_E(|\psi^{SA}\rangle)$ is the Schmidt number, which is defined in~\cite{EisertPRA01}, and $\mathcal{L}_E(|\psi^{SA}\rangle)=\log_2r$. Note that the equality holds if and only if the matrix $M$ is diagonal matrix.

This relation shows that the sum between the entanglement and coherence contained in one subsystem can be not more than the total coherence.
 This relation can be generalized to the mixed states, for any bipartite mixed state $\rho^{SA}$, we have
\begin{equation}\label{eq:E_C_log12}
\max\{\mathcal{L}_C(\rho^S),\mathcal{L}_C(\rho^A)\}+\mathcal{L}_E(\rho^{SA})\leq \mathcal{L}_C(\rho^{SA}).
\end{equation}
Here, $\mathcal{L}_{E}(\rho^{SA})$ is the Schmidt number of mixed state, which is defined as
\begin{equation}
\mathcal{L}_E(\rho^{SA})=\min_{\{p_i, |\psi_i^{SA}\rangle\}}\sum_ip_i\mathcal{L}_E(|\psi_i^{SA}\rangle),
\end{equation}
where the minimum is taken over all pure state decompositions of
$\rho^{SA}=\sum_ip_i|\psi_i\rangle^{SA}\langle\psi_i|$.

In fact, our results~(\ref{eq:super-addi}) and (\ref{eq:E_C_log12}) are also generalize to the multipartite setting. Let $\rho^{SA_1\cdots A_N}$ be a $N+1$-partite states, by the repeated use of the super-additivity, we have
\begin{equation}\label{eq:E_C_log123}
\mathcal{L}_C(\rho^{S})+\sum_{i=1}^{N}\mathcal{L}_C(\rho^{A_i})\leq\mathcal{L}_C(\rho^{A_1\cdots A_N}).
\end{equation}
Combining Eq.~(\ref{eq:E_C_log12}) and Eq.~(\ref{eq:E_C_log123}), we have
\begin{equation}
\mathcal{L}_E(\rho^{S|A_1\cdots A_N})+\sum_{i=1}^{N}\mathcal{L}_C(\rho^{A_i})\leq\mathcal{L}_C(\rho^{A_1\cdots A_N}),
\end{equation}
where $\mathcal{L}_E(\rho^{S|A_1\cdots A_N})$ is the Schmidt number with the bipartite cut $S|A_1\cdots A_N$.

Finally, it is interesting to compare the logarithmic coherence number with the Schmidt number.
We consider quantum-incoherent state, which has the following form
\begin{equation}
\chi^{SA}=\sum_ip_i|i\rangle^S\langle i|\otimes \rho^A_i,
\end{equation}
where $\rho^A_i$ are arbitrary quantum states on $A$, and the states $|i^S\rangle$ belong to the local incoherent basis of $S$~\cite{Chitambarprl16}.
For any quantum-incoherent state, we easily obtain that the Schmidt number is zero, i.e.,
\begin{equation}
\mathcal{L}_E(\rho^{SA})=0.
\end{equation}
But, we can also obtain
\begin{equation}\label{eq:QI_LCN11}
\mathcal{L}_C(\chi^{SA})\leq \sum_ip_i\mathcal{L}_C(\rho^A_i).
\end{equation}
We note that the minimum in $\mathcal{L}_C(\chi^{SA})$ depends only on the pure decomposition of $\rho^A_i$, without loss of generality,
let $\chi^{SA}=\sum_{ij}p_iq_j|i\rangle^S\langle i|\otimes |\psi_{ij}\rangle^A\langle\psi_{ij}|$ be the optimal decomposition of $\chi^{SA}$ belonging to the minimum in Eq.~(\ref{eq:LCN}), we have
\begin{align}\label{eq:QI_LCN12}
\mathcal{L}_C(\chi^{SA})&=\sum_{ij}p_iq_j\mathcal{L}_C(|\psi_{ij}\rangle^A\langle\psi_{ij}|)\nonumber\\
&=\sum_{i}p_i\sum_jq_j\mathcal{L}_C(|\psi_{ij}\rangle^A\langle\psi_{ij}|)\nonumber\\
&\geq \sum_{i}p_i\mathcal{L}_C\left(\sum_jq_j|\psi_{ij}\rangle^A\langle\psi_{ij}|\right)\nonumber\\
&=\sum_{i}p_i\mathcal{L}_C(\rho^A_i).
\end{align}
Combining Eq.~(\ref{eq:QI_LCN11}) and Eq.~(\ref{eq:QI_LCN12}), we have
\begin{equation}\label{eq:QI_LCN13}
\mathcal{L}_C(\chi^{SA})=\sum_ip_i\mathcal{L}_C(\rho^A_i).
\end{equation}

\section{Converting coherence to entanglement}\label{sec:entanglement}
 In this section, using the logarithmic coherence number,
we discuss the relation between the coherence of a mixed state $\rho^S$ in an initial system $S$ with the entanglement generated from $\rho^S$ by attaching an ancilla system $A$ and taking an incoherent operation $\Lambda^{SA}$ on the bipartite system $SA$. Based on different measures, some authors have been investigated as well~\cite{StreltsovPRL15,MaPRL16,RegulaNJP18,KilloranPRL16}.
\begin{proposition}
The entanglement generated from a state $\rho^{S}$ via an incoherent operation $\Lambda^{SA}$ is bounded above by the logarithmic coherence number, i.e.,
\begin{equation}
\mathcal{L}_{C}(\rho^{S})\geq\mathcal{L}_{E}(\Lambda^{SA}(\rho^{S}\otimes |0\rangle^{A}\langle0|)).
\end{equation}
\end{proposition}
\begin{proof}
Let $|0\rangle\langle0|^{A}$ be an incoherent state on $A$, then we have
 \begin{align}
 \mathcal{L}_{C}(\rho^{S})&=\mathcal{L}_{C}(\rho^{S}\otimes |0\rangle\langle0|^{A})\nonumber\\
&\geq \mathcal{L}_{C}(\Lambda^{SA}(\rho^{S}\otimes |0\rangle\langle0|^{A}))\nonumber\\
%&=\mathcal{L}_{C}\left(\sum_{i,n}q_np_i^*(|\psi_{i,n}^*\rangle^{SA}\langle\psi_{i,n}^*|\right)\nonumber\\
&=\sum_k\lambda_k\mathcal{L}_C(|\phi\rangle^{SA})\nonumber\\
&\geq \sum_k\lambda_k\mathcal{L}_E(|\phi\rangle^{SA})\nonumber\\
&= \mathcal{L}_{E}(\Lambda^{SA}(\rho^{S}\otimes |0\rangle\langle0|^{A}),
 \end{align}
 where the second equality comes from the fact that $\Lambda^{SA}(\rho^{S}\otimes |0\rangle\langle0|^{A})=\sum_k\lambda_k|\phi_k\rangle^{SA}\langle\phi_k|$ is an optimal pure states decomposition of $\Lambda^{SA}(\rho^{S}\otimes |0\rangle\langle0|^{A})$ belonging to the minimum in Eq.~(\ref{eq:LCN}), and the second inequality depends on the fact that the coherence rank is greater than or equal to the Schmidt rank.
\end{proof}

From the results in~\cite{KilloranPRL16,ChinPRA17,StreltsovPRL15,RegulaNJP18}, we know that a unitary operation which makes the coherence rank and the Schmidt number equal is given by
\begin{equation}\label{eq:unitary11}
U=\sum_{i=0}^{d-1}\sum_{j=i}^{d-1}|i\rangle^S\langle i|\otimes |i\oplus(j-1)\rangle^A\langle j|,
\end{equation}
where $\oplus$ mens an addition modula $d$. Let
\begin{equation}
|\psi^S\rangle=\sum_{i}\lambda_i|i^S\rangle
\end{equation}
be a pure state on $S$. The unitary operation maps the state $|\psi^S\rangle\otimes |0^A\rangle$ to the state
\begin{equation}
U(|\psi^S\rangle\otimes |0^A\rangle)=\sum_i\lambda_i|i^S\rangle|i^A\rangle.
\end{equation}
Then we easily obtain
\begin{equation}\label{eq:E_C_log11}
\mathcal{L}_{C}(|\psi^S\rangle)=\mathcal{L}_{E}(U(|\psi^S\rangle\otimes |0^A\rangle)).
\end{equation}
Similar to the result~\cite{RegulaNJP18}, we can extend it to the general case of mixed states as follows.
\begin{proposition}
There exists an isometry $W:\mathcal{H}^S\rightarrow \mathcal{H}^S\otimes \mathcal{H}^A$ such that for any state $\rho^S$ on $S$, we have
\begin{equation}
\mathcal{L}_{C}(\rho^S)=\mathcal{L}_{E}(W\rho^SW^\dagger).
\end{equation}
\end{proposition}
\begin{proof}
Let $\{|i\rangle\}$ be an orthonormal basis and $|a\rangle$ be any state in $\mathcal{H}^A$ , one can define
\begin{equation}\label{eq:W_def}
W=\sum_iK_i\otimes|i\rangle\langle 0|.
\end{equation}
Then we have $W^\dagger W=I\otimes |0\rangle\langle 0|$. Note that there exists a unitary operation $U$ such that $W=U(I\otimes|0\rangle\langle 0|)$.
In particular, we take the unitary operation given in~(\ref{eq:unitary11}).
Let $\rho=\sum_i\lambda^*_i|\psi_i^*\rangle\langle\psi_i^*|$ be a decompositions for which the minima in Eq.~(\ref{eq:LCN}) is attained.
Since the operation $I\otimes|0\rangle\langle 0|$ does not effect the Schmidt number, for any state $|\psi_i^*\rangle$, using Eq.~(\ref{eq:E_C_log11}), then we have
\begin{equation}
\mathcal{L}_{C}(|\psi_i^*\rangle)=\mathcal{L}_{E}(W|\psi_i^*\rangle).
\end{equation}
We know that there exists a one-to-one correspondence between the pure states decompositions of $\rho$ and the decompositions of $\rho^\prime=W\rho W^\dagger$ for given $W$, then we obtain $\{\lambda^*_i, W|\psi_i^*\rangle\}$ will form an optimal pure-state decomposition of $\rho^\prime$, and
\begin{align}
\mathcal{L}_{C}(\rho)&=\sum_i\lambda^*_i\mathcal{L}_{C}(|\psi_i^*\rangle)\nonumber\\
&=\sum_i\lambda^*_i\mathcal{L}_{E}(W|\psi_i^*\rangle)\nonumber\\
&=\mathcal{L}_{E}(W\rho^SW^\dagger).
\end{align}
This completes the proof of the proposition.
\end{proof}

\section{Conclusions}\label{sec:conclusion}
We have introduced a new coherence measure of coherence, the logarithmic coherence number, which is generalized from the Schmidt measure and coherence rank. We have shown that the logarithmic coherence number is a proper coherence measure. We have also proved the logarithmic coherence number is additive but not continuous. In particular, we have found that the logarithmic coherence number is computable for a large class of states. We have shown that the logarithmic coherence number satisfies the super-additivity, and obtained the relationship between coherence and entanglement via our presented measures. The results can be also extended to multipartite setting. We have shown that the creation of entanglement
with bipartite incoherent operations is bounded by the logarithmic coherence number of the initial system during the process. Some interesting results are given. We hope this measure of coherence will improve the understanding of quantum resource theory.
\section{Acknowledgments}
Z. Xi is supported by the National Natural Science Foundation of China (Grant Nos. 61671280,11531009, and 11771009), and by the Natural
Science Basic Research Plan in Shaanxi Province of China
(No. 2017KJXX-92),
and by the Funded Projects for the Academic Leaders and Academic Backbones, Shannxi Normal University (16QNGG013).

\end{document}